\pdfoutput=1
\documentclass{aptpub}
\usepackage{amsmath}
\usepackage{amssymb}
\usepackage{color}
\usepackage{enumerate}
\usepackage{graphicx}
\usepackage{caption}
\usepackage{subcaption}
\usepackage{bbm}
\usepackage{lmodern}
\usepackage{algorithm, algorithmic}
\usepackage{setspace}
\newcommand{\NN}[0]{\mathbb{N}}
\newcommand{\EE}[0]{\mathbb{E}}
\newcommand{\PP}[0]{\mathbb{P}}

\newcommand{\RR}[0]{\mathbb{R}}
\newcommand{\cW}[0]{\mathcal{W}}
\newcommand{\ind}[0]{\mathbbm{1}}
\newcommand{\bu}[0]{\mathbf{u}}
\newcommand{\bv}{\mathbf{v}}
\newcommand{\bx}[0]{\boldsymbol{x}}

\newcommand{\limt}[0]{\lim_{t \rightarrow \infty}}
\newcommand{\wtilde}[1]{\widetilde{#1}}
\newcommand{\1}{\mathbbm{1}}
\numberwithin{equation}{section}

\newcommand{\edit}[1]{\color{black}{#1}}
\newcommand{\qedsymbol}[0]{\edit{\square}}

\authornames{GRIFFIN,A., JENKINS,P.A., ROBERTS,G.O., SPENCER,S.E.F.} 
\shorttitle{Simulating QSDs on Reducible State Spaces} 

\begin{document}

\title{Simulation from quasi-stationary distributions\\on reducible state spaces} 

\authorone[University of Warwick]{A. Griffin}
\authorone[University of Warwick]{P.A. Jenkins}
\authorone[University of Warwick]{G.O. Roberts}
\authorone[University of Warwick]{S.E.F Spencer}
\addressone{Dept of Statistics, University of Warwick, Coventry, CV4 7AL, UK}

\begin{abstract}
	\edit{
	Quasi-stationary distributions (QSDs)
	 arise from stochastic processes that exhibit transient equilibrium behaviour on the way to absorption. 
	 QSDs are often mathematically intractable and even drawing samples from them is not straightforward. 
	In this paper the framework of Sequential Monte Carlo samplers is utilized to simulate QSDs and several novel resampling techniques are proposed to accommodate models with reducible state spaces, with particular focus on preserving particle diversity on discrete spaces. Finally an approach is considered to estimate eigenvalues associated with QSDs, such as the decay parameter.}
\end{abstract}

\keywords{Quasi-stationary distributions; limiting conditional distributions; simulation; resampling methods; Sequential Monte Carlo} 

\ams{60J27}{62G09} 
\section{Introduction}


Quasi-stationary distributions (QSDs) and limiting conditional distributions (LCDs) arise from processes that exhibit temporary equilibrium behaviour before hitting an absorbing state. They appear commonly in population processes \cite{meleard2012}, epidemic models \cite{clancy2003, nasell1999, neal2014} and models in population genetics \cite{lambert2008, wri:1931}. 
\edit{Recently there has been a significant increase in interest in the use of quasi-stationarity in the context of simulation using sequential Monte Carlo methods. For example \cite{pollock2016scalable} uses quasi-stationary simulation for Bayesian inference with big data sets, an entirely different motivation from ours in this paper.}
\edit{The mathematical analysis of QSDs and LCDs are} often extremely challenging due to the fact that even for relatively simple processes the distributions do not always exist; if they do exist they may not be unique; and there is typically no closed form expression available for the distribution function. 
Furthermore, it is reasonably challenging to simulate from QSDs and LCDs due to the fact that the absorption event becomes increasingly likely through time but, to be representative of the QSD, any sample paths must survive long enough to have `forgotten' their starting state.

In this paper we focus on the problem of simulating from the LCD of a stochastic process on a countable state space. In particular we make use of the Sequential Monte Carlo (SMC) sampler \cite{delmoral2006} 
and discuss several novel resampling steps which make SMC sampling for QSDs and LCDs more efficient. In particular we address the difficulties that arise when sampling from processes that have a reducible state space.

\subsection{Quasi-stationary and limiting conditional distributions}\label{subsec: qsd}
Consider a continuous-time stochastic process $(X(t))_{t\geq0}$ evolving on state space $\Omega = S \cup \{0\}$, where $S$ contains transient states and $0$ is the identification of all absorbing states. Denote the transition rate matrix by $\widetilde{Q}$ and denote its restriction to the transient states $S$ by $Q$.

\begin{defn}\label{def: reduce}
A process \edit{$(X(t))_{t \geq 0}$} (or the state space $S$ over which it evolves) is said to be \emph{irreducible} if for every $i,j \in S$ there exists $t \in (0,\infty)$ such that
$P_{ij}(t): = \PP[X(t) = j | X(0) = i] > 0 $
and $P_{ji}(t) > 0$; every state in $S$ can be accessed from every other state. A process and its state space are considered \emph{reducible} if not irreducible.
\end{defn}

\begin{defn}\label{def: qsd}
	A proper probability distribution $\bu = (u_j: j \in S)$ is said to be a \emph{quasi-stationary distribution} (QSD) for a Markov process $(X(t))_{t \geq 0}$ if, for every $t \geq 0$,
	\begin{align*}
		u_j = \PP[X(t) = j | X(0) \sim \bu, X(t) \in S], \qquad j \in S.
	\end{align*}
	That is, conditional on the process not being absorbed, the distribution of \edit{$(X(t))_{t \geq 0}$} started from initial distribution $\bu$ is time-invariant.
\end{defn}
		
\begin{defn}\label{def: lcd}
	Given a proper probability distribution $\nu$ on $S$, a proper probability distribution $\bu$ on $S$ is said to be a $\bv$-\emph{limiting conditional distribution} ($\nu$-LCD) for the Markov process $(X(t))_{t\geq 0}$ if, for each $j \in S$,
	\begin{align*}
		u_j = \limt \PP[X(t) = j | X(0) \sim \nu, X(t) \in S].
	\end{align*}
	If $\nu$ is a point mass at state $i \in S$, then we may refer to this $\nu$-LCD as an $i$-LCD.
\end{defn}

Every QSD $\bu$ is a $\bu$-LCD, and every LCD is a QSD \cite[Prop 1]{meleard2012}. We outline here some results regarding the existence and uniqueness of QSDs. 

\begin{prop}\label{prop: eu}
When $X = (X(t))_{t \geq 0}$ is finite and irreducible, there exists a unique LCD (which is independent of initial distribution $\nu$) and there exists a unique QSD. Moreover, these two distributions are equal. \cite[Thm 3]{darroch1967}. 
\end{prop}

\edit{
\begin{remark}
  \begin{itemize}
    \item When $S$ is finite but reducible, it can be seen that the $\nu$-LCD depends on the communicating classes in which the support of $\nu$ is contained. \cite[Thm 7]{vandoorn2008}.
    \item For countable $S$, existence is not assured even with certain absorption, and uniqueness does not hold in general. For example, if the transition probabilities $P_{ij}(t)$ of an absorbing process $X(t)$ converge to zero only polynomially, then no QSD exists \cite[Thm 6]{vandoorn2013}.
  \end{itemize}
\end{remark}
}

Some \edit{works have} already been done into simulating QSDs using different methods. Groisman \& Jonckheere \cite{groisman2012} consider a renewal process where, instead of being absorbed, particles are redrawn immediately from a given distribution over the transient states. 
A supercritical multitype branching process is proposed, which can be used to simulate the QSD through the use of a Kesten-Stigum theorem. Another paper by \edit{Blanchet, Glynn and Zheng \cite{blanchet2013}} uses a different approach. A single particle is simulated to absorption, then a new particle is restarted from a location drawn from the distribution of states visited by the original particle, weighted by occupation time. They prove that as more particles are simulated, the distribution of occupation times converges to the LCD. \edit{However, these methods assume irreducibility, a condition we wish to go beyond in this paper.} In this work, we provide an alternative and potentially more efficient method, particularly for processes with reducible state spaces. 
The methods discussed in the rest of this paper hope to improve on the above papers by making more strategic choices of resampling in order to improve the variance of the estimators in the case of reducible state spaces (which are not discussed in the above) \edit{whilst making use of more computationally viable techniques} which do not require unbounded population sizes. Moreover, it allows the use of known results in Sequential Monte Carlo in order to prove convergence of our simulations in countable and potentially reducible state spaces.

\subsection{Motivating example: pure death process}\label{subsubsec: motive}
To illustrate the difficulties involved in simulating QSDs for processes on reducible state spaces, consider a pure death process $(X(t))_{t \geq 0}$ on $\{0,1,\dots,L\}$ which evolves according to a given sequence of death rates $\{\delta_i : i = 1, \dots, L\}$ shown in Figure \ref{fig: d30}.
The process jumps from state $i$ to state $i-1$ after an exponentially distributed waiting time with rate $\delta_i$. Once the process hits state $0$ it stays there for all future time -- it is \emph{absorbed}. 
This process inhabits a reducible state space in which each state is a communicating class. For this example, 
a single \edit{$L$-LCD exists} and can be calculated by solving the left eigenvector problem $\bu^{T}Q = -\alpha \bu^{T}$
 . In the pure death process, this simplifies to solving
\edit{
\begin{align*}
		- \min(\delta_j : 1 \leq j \leq L) u_{i} =  \delta_{i+1} u_{i+1} - \delta_i u_i  \qquad i=1, \dots, L
\end{align*}}
where we stipulate $\sum_{i=1}^L u_i = 1$ to ensure uniqueness. This will be expanded on in Section \ref{puredeath}.
We compare two approaches for simulating the LCD: (1) a rejection sampler and (2) the SMC sampler with two new resampling methods introduced in this paper: combine-split resampling and regional resampling. 
	
Figure \ref{fig: mn d30} shows the empirical distribution produced by each method. Under standard rejection sampling, the accepted particles don't provide a good estimate of the LCD starting from state 30, due to the large number of rejections (only 52 out of 3000 sample paths are not rejected), and the absence of particles present in states $\{20,\dots,30\}$. However in Figure \ref{fig: rr d30} the SMC sampler, which makes use of combine-split resampling (Section \ref{combine}) and regional resampling (Section \ref{regional}) performs much better.
	
\begin{figure}
	\centering
	\begin{subfigure}[b]{0.31\textwidth}
		\includegraphics[width=\textwidth]{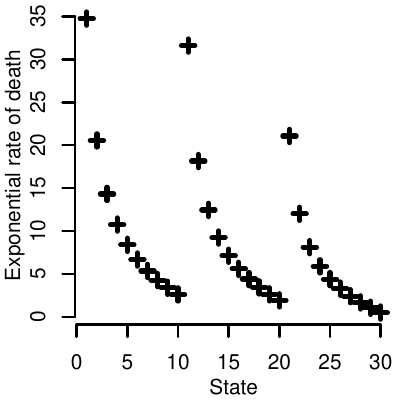}
		\caption{Death Rates of Process}
		\label{fig: true d30}
	\end{subfigure}\quad
	\begin{subfigure}[b]{0.31\textwidth}
		\includegraphics[width=\textwidth,page=2]{d31mn_v3.pdf}
		\caption{Rejection sampler}
		\label{fig: mn d30}
	\end{subfigure}\quad
	\begin{subfigure}[b]{0.31\textwidth}
		\includegraphics[width=\textwidth,page=2]{d31rr_v3.pdf}
		\caption{SMC sampler}
		\label{fig: rr d30}
	\end{subfigure}
	\caption{Simulated LCD for a pure death process using rejection sampler and SMC sampler. Both used $M=3000$ particles and simulated to $T_{\text{end}}=30$. The SMC sampler used combine-split and regional resampling with 1000 particles in each of the regions $S_1=\{1, \dots, 10\}$, $S_2=\{11, \dots, 20\}$, $S_3 = \{21, \dots, 30\}$, $T_{\text{b}} = 0$, $T_{\text{step}} = 5$, $T_{d}=30$.}
	\label{fig: d30}
\end{figure}	
In Section \ref{SMC} we introduce the SMC sampler as a method of simulating LCDs, discuss common resampling methods and apply SMC sampling to the Wright-Fisher model from population genetics. For an introduction to the SMC framework\edit{, see \cite{delmoral2004}.} Section \ref{combine} introduces a \emph{combine-split resampling} step to the SMC algorithm, which helps to avoid particle degeneracy when the state space is discrete. An application to the linear birth-death process shows that a combine-split resampling step can help the SMC to explore the tails of the distribution. 
Section \ref{regional} introduces a \emph{regional resampling} step to the SMC algorithm, which may be needed if the state space is reducible. This is illustrated using a pure death process and a transient immunity process -- a model for an emerging epidemic on a large population. In Section \ref{stopping} we discuss how to use resampling based on stopping-times to avoid a pitfall of the SMC sampler approach. 

\section{Sequential Monte Carlo samplers}\label{SMC}
\subsection{Definition} 
\emph{Sequential Monte Carlo Samplers} \cite{delmoral2006} 
provide a means to sample from a sequence of target measures $\{\pi_n: 1 \leq n \leq N \}$ over some common measurable space $(\edit{E}, \mathcal{F})$. Typically the target measures are only known up to some normalization constant $Z_n$ and so it common to work with the unnormalized measure $\edit{\gamma_n} = Z_n \pi_n$.

Sampling is done using a sequence of proposal measures $\{ \eta_n: 1 \leq n \leq N\}$ on $\edit{E}$. 
Given an initial proposal measure $\eta_1$, we construct subsequent proposal measures using
 $\eta_n(x_n) = \int_{\edit{E}} \eta_{n-1}(x_{n-1}) K_n(x_{n-1},x_n) dx_{n-1}$ 
for some sequence of Markov kernels $K_n: \edit{E} \times \mathcal{F} \rightarrow [0,1]$. 
Under naive importance sampling we would give each particle the unnormalized importance weight $w_n(x_n) = \edit{\gamma_n}(x_n)/\eta_n(x_n)$. However, such proposal distributions $\eta_n$ become very difficult to compute pointwise as $n$ increases, particularly if $\edit{E}$ is high-dimensional.
To tackle this we use the following SMC sampler as described in \cite{delmoral2006}. We define a sequence of artificial backwards-in-time Markov kernels $L_{n-1}(x_n,x_{n-1})$ then perform importance sampling using joint proposal distributions $\eta_n(\bx_{1:n})$ (where $\bx_{1:n} = (x_1, \dots, x_n)$) to estimate an artificial target joint distribution $\wtilde{\pi}_n(\bx_{1:n})$ on $\edit{E^n}$ defined by
\edit{
	$\wtilde{\pi}_n(\bx_{1:n}) := Z_n^{-1} \edit{\wtilde{\gamma}_n}(\bx_{1:n})$
	where
	$\edit{\wtilde{\gamma}_n}(\bx_{1:n}) := \edit{\gamma_n}(x_n) \prod_{k=1}^{n-1} L_k(x_{k+1}, x_k)$.
}
Our final target distribution $\pi_n(x_n)$ is a marginal of our artificial target by construction.
Assuming we can evaluate $\eta_1(x_1)$ and $\edit{\gamma_1}(x_1)$ to obtain unnormalized weights $w_1(x_1)=\frac{\edit{\gamma_1}(x_1)}{\eta_1(x_1)}$ then for each time-point $n$, we move the particles forward according to $K_n(x_{n-1},x_n)$. We then use importance sampling to approximate our artificial target $\wtilde{\pi}_n(\bx_{1:n})$ which gives unnormalized importance weights
\begin{align}\label{eqn: wn}
	w_n(\bx_{1:n}) &= w_{n-1}(\bx_{1:n-1}) \wtilde{w}_n(x_{n-1},x_n) \quad \text{for } n > 1 \\
	 \text{ where } \qquad \wtilde{w}_n(x_{n-1}, x_n) &= \frac{ \edit{\gamma_n}(x_n) L_{n-1}(x_n, x_{n-1})}{\edit{\gamma_{n-1}}(x_{n-1}) K_n(x_{n-1},x_n)} \nonumber.
\end{align}
We normalise the importance weights to get $W_n(\bx_{1:n})$. These weights can then be used to generate samples from the marginal distribution of interest $\pi_n(x_n)$.

One can combine the SMC sampler with a resampling scheme to reduce the effect of particle weight degeneracy, as typically, \edit{one will end up with particles of very high weight or very low weight, arising from an increase in the variance of the particle weights proven in \cite{kong1994}.} Such resampling schemes will be discussed in Section \ref{subsec: resample}. 

%
%
\subsection{Simulating LCDs using SMC samplers}
To implement the SMC sampler in order to simulate from a $\nu$-LCD, we would ideally take the $\nu$-LCD itself to be the target distribution. However, in general we will not be able to compute even the unnormalized density of the LCD, and so we cannot compute importance weights within an SMC sampler scheme. Instead, we use the time marginal of the process in question conditional on non-absorption given by $\pi_T(\cdot) = \PP[X(T) \in \cdot |  X(T) \in S, X(0)\sim \nu ]$, which converges to the true $\nu$-LCD as $T$ gets large. It is the time marginal $\pi_T$ which we will attempt to simulate using an SMC sampler. \edit{As such, the standard finite-time biases are accounted for, although one should note that there is generally bias in finite-particle simulations \cite{delmoral2004}.}

Define an increasing sequence of time points $\{ t_n: n=1, \dots, N \}$ with $t_1=0$ and $t_N=T$, and set the initial proposal distribution $\eta_1$ to be $\nu$. When simulating analytically intractable LCDs we cannot even work with an unnormalized target $\edit{\gamma_n}(x)$, and so we might as well set $\pi_n(x) = \edit{\gamma_n}(x)$ (i.e. $Z_n=1$).
We construct the sequence of proposal distributions $\eta_n(\cdot) = \PP[X(t_n) \in \cdot | X(t_1) \sim \eta_1]$ using the Markov transition kernel $K_n(x_{n-1},x_n) = \PP[X(t_n) = x_n | X(t_{n-1}) = x_{n-1}]$ and the backward kernel
\begin{align*}
	L_{n-1}(x_n, x_{n-1}) &= \frac{\PP[X(t_{n-1})=x_{n-1}]}{\PP[X(t_n)=x_n]} \PP[X(t_n)=x_n | X(t_{n-1})=x_{n-1}].
\end{align*}
This precisely matches the optimal choice of backward kernel given by \cite{delmoral2006}, which minimizes the variance of the unnormalized importance weights $w(\bx_{1:n})$.
The forward kernel corresponds to the ``bootstrap'' choice which ignores the information that absorption does not occur. Superior choices may be possible for specific models, but this is a flexible and easily implementable kernel, and one for which the optimal backward kernel is available. Expressing our simulation problem as a generic SMC algorithm gives us access to standard convergence results; see for example \cite{delmoral2004}.


\begin{prop}\label{prop}
	If $\nu$ is supported on $S$, then 
	\edit{$w_n(x_n) = \edit{{\PP[X(t_n) \in S]}^{-1}}$ for $x_n \in S$ and zero otherwise.}
\end{prop}
\begin{proof}
	Since $\eta_1=\nu$, for $n=1$ we have that $w_1(x_1)=\frac{\pi_1(x_1)}{\eta_1(x_1)}=\frac{\PP[X(t_1)=x|X(t_1)\in S, X(t_1)\sim\nu]}{\eta_1(x)}=1$ and so the particles begin with equal weights. For $n>1$, if we substitute the expressions for $L_n, K_n, \eta_n$ and $\pi_n$ into the incremental weight $\wtilde{w}_n$ (suppressing the conditioning on $X(t_1)\sim\nu$ for brevity) we see that
	\begin{align*}
		\wtilde{w}_n(x_{n-1}, x_n) &= \frac{ \pi_n(x_n) L_{n-1}(x_n, x_{n-1})}{\pi_{n-1}(x_{n-1}) K_n(x_{n-1},x_n)} \nonumber\\
		&= \frac{ \pi_n(x_n) }{ \pi_{n-1}(x_{n-1})} \frac{\PP[X(t_{n-1})=x_{n-1}]}{\PP[X(t_n)=x_n]} \frac{\PP[X(t_n)=x_n | X(t_{n-1})=x_{n-1}]}{\PP[X(t_n)=x_n | X(t_{n-1})=x_{n-1}]} \nonumber\\
		&= \frac{ \PP[X(t_n) = x_n | X(t_n) \in S] }{\PP[X(t_n)=x_n]}\frac{\PP[X(t_{n-1})=x_{n-1}]}{\PP[X(t_{n-1}) = x_{n-1} | X(t_{n-1}) \in S] }. \nonumber 
	\end{align*}
	Substituting this incremental weight into the full weight \eqref{eqn: wn} gives a telescoping product which reduces to
	\begin{align*}
		w_n(x_n) &= \frac{ \PP[X(t_n) = x_n | X(t_n) \in S] }{\PP[X(t_n)=x_n]}w_1(x_1).
	\end{align*}
	\edit{The result follows immediately by the definition of conditional probability.}
	$\qedsymbol$
\end{proof}

Proposition \ref{prop} shows that simulating the $\nu$-LCD via SMC sampling works in a similar way to rejection sampling, where only the non-absorbed particles are considered. The existence of a LCD requires the certain absorption of each particle in a finite amount of time and so there remains a problem balancing the approximation of the LCD (improved by increasing $T$) and particle depletion (worsened by increasing $T$). However setting our algorithm within an SMC framework allows us to draw on existing tools to prevent particle depletion, such as particle resampling.

\subsection{Particle resampling methods}\label{subsec: resample}
Without a resampling scheme, standard SMC methods suffer from what is referred to in the literature as \emph{particle weight degeneracy} where, when a simulation is run long enough, one will end up with a single particle with nearly all the weight, and many low-weight particles. This results in a low effective sample size and a poor estimate of the target distribution. 

To implement a resampling scheme, we define a sequence of resampling timepoints $\{ \tau_1, \tau_2, \dots, \tau_k\}$, either deterministically or drawn randomly according to some given distribution. At each resampling timepoint, we redraw $M'$ particles with replacement from a pool of $M$ particles $\{ (X_j, W_j): j=1, \dots, M\}$ with normalized weights. Examples of these include \emph{Multinomial Resampling} \cite{liu1995} and \emph{Residual Resampling} \cite{liu1998}.

\begin{remark}
One should note that in 
these resampling methods, one will typically draw many particles from the high-weighted locations, and very few if any from low weighted locations. For the purposes of exploring the tails of the distribution, this is not particularly desirable. 
Specifically, under standard resampling methods, one would expect, in order to look at the top $1 \%$ of the distribution, one would require many more than 100 equally weighted particles, due to the resampling having a very low chance of selecting such particles. To this end, we developed resampling methods which maintain high levels of particle diversity.
\end{remark}

\subsection{Particle refilling}\label{subsec: refilling}
To obtain approximations of LCDs using SMC samplers, we have shown that we need only simulate the unconditioned process and give uniform non-zero weight to all the non-absorbed particles. A simple approach aimed at maintaining particle diversity, which we refer to as \emph{particle refilling}, is to resample only those particles which have been absorbed. If $A$ particles have been absorbed then we replace these particles with a sample drawn from the non-absorbed particles using one of the existing resampling mechanisms described above, with $M'=A$. 
Since each non-absorbed particle is a draw from the process conditional on non-absorption, one can intuitively see that \edit{there is no gain by resampling such particles}. If applied to all $M$ particles, \edit{any valid resampling method} can potentially replace non-absorbed particles, which reduces particle diversity and hinders the estimation of the tails of a LCD. Particle refilling maintains particle diversity, since there is no chance of removing any non-absorbed particles. 

We next show that particle refilling is valid in the sense that a properly weighted sample is still properly weighted after refilling. 
\edit{
\begin{defn}\label{def: proper}
A set of weighted random samples $\{(X_j,w_j):1\leq j \leq M\}$ is called \emph{proper} with respect to $\pi$ if for any square integrable function $h(\cdot)$ we have 
$\EE[h(X_j)w_j]=c \EE_{\pi}[h(X)]$  for $j=1,\dots,M$,
where $c$ is a normalising constant common to all $M$ samples. (Chapter 10 of \cite{doucet2001})
\end{defn}
}
\begin{prop}\label{refillingproper}
Given a properly weighted sample $\{(X_j,w_j):1\leq j \leq M\}$ as in Definition \ref{def: proper} and a resampling method that produces properly weighted samples then particle refilling
produces properly weighted samples.
\end{prop}
\begin{proof}
Let $(X_j',w_j')$ denote the particle location and weight after particle refilling. Then conditioning on $X_j=0$ (or equivalently $w_j=0$) yields $\EE[h(X_j')w_j']=\EE[h(X_K)w_K]\PP(X_j=0)+\EE[h(X_j)w_j]\PP(X_j\neq 0)=c\EE_{\pi}[h(X)]$, where $K$ is the index of the resampled particle randomly chosen from the non-zero weight particles using the given resampling method.  $\qedsymbol$
\end{proof}

\subsection{Taking multiple samples}
When simulating LCDs using SMC samplers, we can take advantage of the fact that after a suitable burn-in period, every target distribution $\pi_n$ is an approximation of the true LCD. As a result, we borrow ideas from Markov Chain Monte Carlo and draw samples from many timepoints after the SMC sampler has reached stationarity and not just the final timepoint. We adopt a burn-in period $T_{b}$, during which the samples are discarded. To reduce auto-correlation between samples we also incorporate thinning, in the form of a delay $T_d$ between sampling times. 
	
\subsection{Example: Wright-Fisher model}
In order to demonstrate the SMC sampler, we apply it to simulate the QSD of the discrete-time Wright-Fisher process. Population genetics forms an important application area for QSDs; see \cite{lambert2008} for an overview covering both the Wright-Fisher model and QSDs. In this context the absorbing states correspond to the loss or fixation of a mutation in a large population, yet genetic variation in a modern population is---by definition---the collection of those mutations for which absorption has not occurred in the time up to the present day.

Consider a population of $D$ haploid individuals each with an allelic type from a set of 2 types denoted $\{1,2\}$. At each timepoint $n$, we generate a new set of $D$ offspring, which will form the population of the next generation. If the allelic types confer no advantage, then each offspring independently chooses a parent uniformly at random and adopts the allelic type of the parent. More generally, each allelic type $k$ is assigned a \emph{selection coefficient} $s_k \geq 0$ where each offspring selects a particular parent with type $k$ with probability proportional to $s_k + 1$. This \emph{Wright-Fisher with selection} process \cite{etheridge2011} evolves over a finite state space of 
$2^D$ states.
%

\subsubsection{2-type Wright-Fisher model.}
Consider the LCD for the Wright-Fisher process with selection conditioning on the event that both types remain in the population. One can use results from \cite{vandoorn2013} to find the true LCD as a left eigenvector of the transition matrix, however for large population sizes, this is numerically demanding. Here we use a small enough population to generate the true LCD for illustrative purposes. 
	
We compare the SMC model with multinomial refilling to the basic rejection resampling algorithm. Here we set resampling to occur every 5 timepoints; henceforth, this interval will be given by $T_{\text{step}}$. In Figure \ref{fig: SMC wf2 parts}, the main problem with rejection sampling becomes evident: the number of accepted particles decreases rapidly through time. In comparison, the SMC sampler replenishes the particles at each resampling step. Figure \ref{fig: SMC wf2 hist} shows the estimated QSD of the number of individuals of type 1 from the two methods.
\begin{figure}
	\centering
	\includegraphics[width=0.45\textwidth]{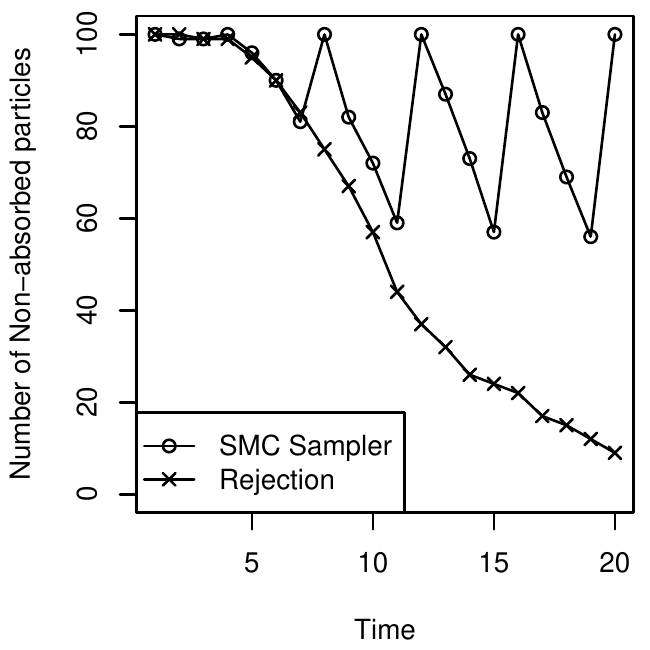}
	\caption{Comparison of the number of non-absorbed particles in the SMC sampler and rejection sampler; $D=20$, $\mathbf{s}= (0, 0.1)$, $M=100$.}
	\label{fig: SMC wf2 parts}
\end{figure}
\begin{figure}
	\centering
	\begin{subfigure}[b]{0.45\textwidth}
		\includegraphics[width=\textwidth]{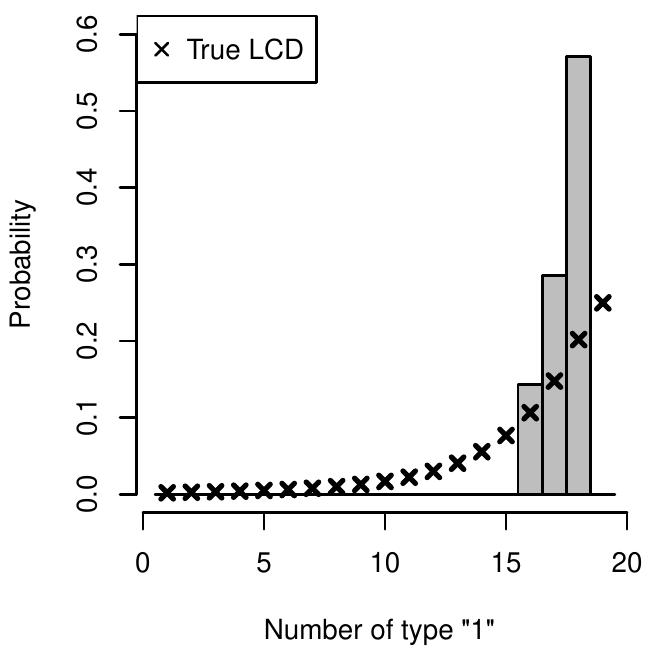}
		\caption{Rejection sampler}
		\label{fig: rejwf}
	\end{subfigure}\quad
	\begin{subfigure}[b]{0.45\textwidth}
		\includegraphics[width=\textwidth]{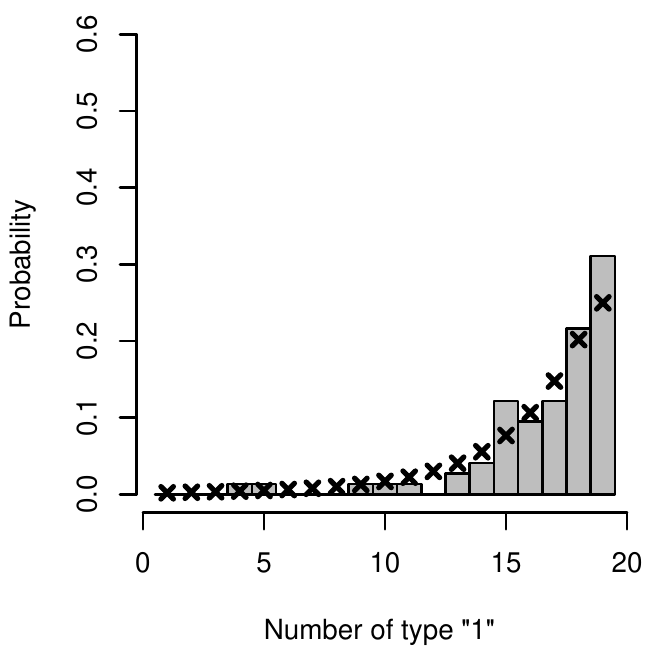}
		\caption{SMC sampler}
		\label{fig: smcwf}
	\end{subfigure}
	\caption{LCD for a Wright-Fisher process with selection, simulated using a rejection sampler and an SMC sampler; $D=20$, $M=100$, Simulation time $T_{\text{end}}=20$, Resampling interval $T_{\text{step}} = 5$, $\mathbf{s}=(0, 0.1)$.}
	\label{fig: SMC wf2 hist}
\end{figure}
\section{Combine-split resampling}\label{combine}
\subsection{Idea and algorithm}
Traditional SMC resampling schemes are designed for particles evolving in continuous state spaces where no single point has strictly positive mass. In discrete state spaces it is likely that several particles share the same location. We propose the \emph{Combine-Split} resampling method which redistributes particles within the space without moving any of the weight, with the aim of improving the effective sample size.

Combine-split resampling comprises 3 steps. First, at each state $s \in S$ we combine together all the particles which are at that location $s$ into a single particle and give it the combined weight of all the particles that were sitting there. Next, all non-assigned particles (those lost in the combining and the absorbed particles) are distributed amongst the locations with non-zero weight according to some chosen distribution, and are assigned temporary weight zero. Finally all of the particles now residing at a given location are given weight equal to the total weight at that location divided by the number of particles there. The combine-split resampling algorithm is given in Algorithm \ref{alg: comb split}. The weight at each state remains constant during combine-split resampling; the particles are simply redistributed amongst those states.

\begin{prop}\label{combineSplitProper}
Given a properly weighted sample \edit{as in Definition \ref{def: proper}} $\{(X_j,w_j):1\leq j\leq M\}$, combine-split resampling produces a properly weighted sample.
\end{prop}
\begin{proof}
Let $\{(X_j',w_j')\}$ denote the locations and weights after combine-split resampling. Since combine-split does not change total weight at any location, we must have $\EE[h(X_j')w_j']=\EE[h(X_j)w_j]=c\EE_\pi[h(X)]$, so the new sample is properly weighted. $\qedsymbol$
\end{proof}
\begin{spacing}{0.8}
\begin{algorithm}[ht]
	\caption{Combine-split resampling}
	\label{alg: comb split}
	\begin{algorithmic}[1]
		\REQUIRE $M \geq 1$, $\{(X_i, W_i): 1 \leq i \leq M \}$ normalized weighted particles.
		\STATE For each $s \in S$ let $a_s = \sum_{i=1}^M W_i \ind_{\{X_i = s\}}$ and $S^*=\{ s \in S : a_s > 0\}$.
		\STATE For each $s \in S^*$ let one particle $X_i = s$ have weight $a_s$. Give all other particles zero weight.
		\STATE For each zero-weight particle, draw a new position from some specified distribution over $S^*$. \label{line:specified}
		\STATE For each $s \in S^*$ let $N_s = |\{i: X_i = s\}|$.
		\STATE For $i = 1, \dots, M$ set $W_i = a_s/N_s$ where $s = X_i$.
		\RETURN $\{(X_i, W_i): 1 \leq i \leq M \}$ normalized weighted particles
	\end{algorithmic}
\end{algorithm}
\end{spacing}
\renewcommand{\arraystretch}{1.0}
Combine-split resampling bears some superficial resemblance to existing resampling mechanisms. It has long been recognized that particles can be resampled according to criteria other than their importance weight, such as the \emph{lookahead} methods reviewed in \cite{linetal2013}. These attempt to predict the utility of a particle to future target distributions via additional Monte Carlo simulation. Our (as-yet unspecified) distribution in \edit{step} \ref{line:specified} of Algorithm \ref{alg: comb split} could be used to do a similar job, though in a non-random manner. 
\edit{However, a key difference with standard resampling algorithms is that, because particle diversity is at such a premium, here resampling is focused on the support $S^*$ of the particles rather than the particles themselves.}
Note that standard lookahead methods, and indeed all of the other resampling mechanisms described in Section \ref{subsec: resample}, can lose particle locations and therefore reduce the diversity of the particles, whereas our method guarantees that no locations are lost. \edit{Since this is a valid resampling method in that its output is properly weighted, it inherits many convergence properties of SMC \cite{douc2008}. Additionally, using a deterministic resampling timepoint sequence allows very easy parallelization (through, for example, MapReplace or Pregel) of this SMC sampler with combine-split resampling, since particles interact only at resampling timepoints.}

\subsection{Example: Combine-split resampling step}
	Suppose that there are eight particles with locations and weights as given in rows 2 and 3 respectively of Table \ref{tab: cs1}, of which two have been absorbed into state zero. The combine step (rows 4 and 5) moves the weight at locations a, b and c to a single particle at each location, leaving 3 extra particles to reallocate (5 in total). Suppose that we reallocate these 5 particles by drawing uniformly at random from the 3 locations, giving $(a,b,b,c,c)$, as seen in rows 6 and 7. Finally, the \emph{Split} step equalises the weight at each location, as shown in rows 8 and 9.
	\begin{spacing}{0.8}
	\begin{table}
		\begin{center}
			\begin{tabular}{cccccccccc}\hline
				Particle name& $X_1$ & $X_2$ & $X_3$ & $X_4$ & $X_5$ & $X_6$ & $X_7$ & $X_8$\\\hline
				Initial Location & a & a & a & b & b & c & 0 & 0\\
				Initial Weight & 1 & 1 & 2 & 1 & 4 & 2 & 0 & 0\\\hline
				Combined Location & a & 0 & 0 & b & 0 & c & 0 & 0 \\
				Combined Weight & 4 & 0 & 0 & 5 & 0 & 2 & 0 & 0\\\hline
				Reallocated Location & a & a & b & b & b & c & c & c\\
				Reallocated Weight & 4 & 0 & 0 & 5 & 0 & 2 & 0 & 0\\\hline
				Split Location & a & a & b & b & b & c & c & c\\
				Split Weight & 2 & 2 & 5/3 & 5/3 & 5/3 & 2/3 & 2/3 & 2/3\\\hline
			\end{tabular}
		\end{center}
		\caption{\label{tab: cs1}Particles' locations and weights during combine-split resampling step.}				
	\end{table}
\end{spacing}
\subsection{Example: Linear birth-death process}
The linear birth-death process $\{I(t):t\geq 0\}$ is a continuous-time Markov process on $\NN_0$ which represents the size of a population of individuals subject to births and deaths. Individuals give birth independently at points of a Poisson Process with rate $\beta$ and live for an exponentially-distributed lifetime with rate $\gamma$. Once the population size hits zero no more births are possible and so zero is an absorbing state. The transient states $S=\{1, 2, \dots\}$ form a single communicating class. At least one QSD exists if $\gamma > \beta > 0$ and there is a unique $\nu$-LCD for any initial distribution $\nu$ with finite mean: the geometric distribution with parameter $\beta/\gamma$ \cite{vandoorn2013}.

In this example we compare the true LCD for the linear birth-death process with the simulated LCDs produced using an SMC sampler with multinomial refilling and with combine-split resampling. In the combine-split resampling step, the zero-weight particles were reallocated to locations drawn uniformly at random from the existing locations, which sends more particles to the tail than reallocating proportionally to the weights. Figure \ref{fig: cs bdp} shows the true and estimated LCDs. It is clear that the SMC with combine-split resampling reaches further into the tail than SMC with multinomial refilling, and hence matches the true LCD more closely. 
In Figure \ref{fig: mean bdp}, we see the cumulative mean of the particles observed at each resampling step. The cumulative mean under combine-split is consistently higher than under multinomial refilling and closer to the true mean. This suggests faster convergence to the true LCD and a reduction in the finite sample bias inherent in particle approximation 
methods.

\begin{figure}
	\centering
	\begin{subfigure}[b]{0.45\textwidth}
		\includegraphics[width=\textwidth]{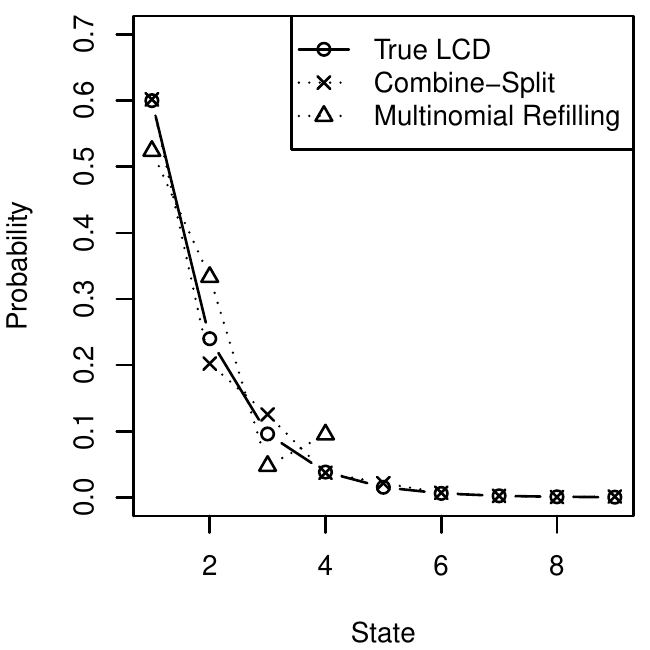}
		\caption{True and estimated LCDs}
		\label{fig: cs bdp}
	\end{subfigure}\quad
	\begin{subfigure}[b]{0.45\textwidth}
		\includegraphics[width=\textwidth]{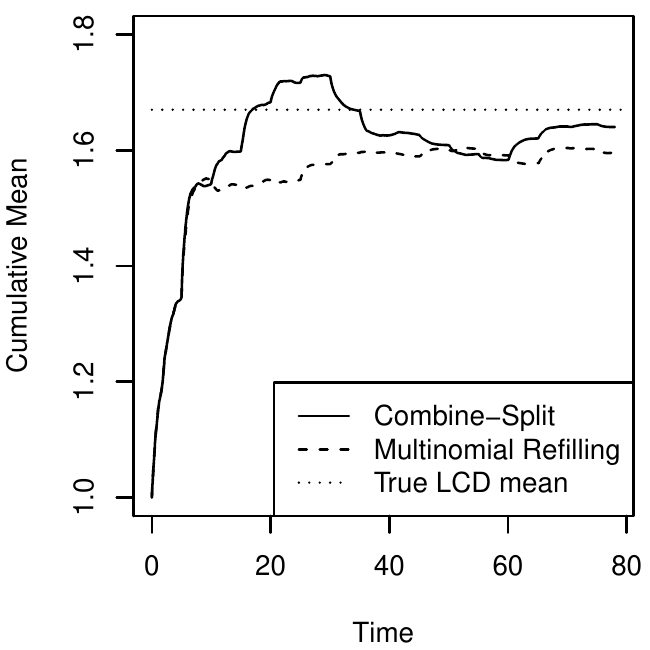}
		\caption{Cumulative mean of particles}
		\label{fig: mean bdp}
	\end{subfigure}
	\caption{The LCD for a linear birth-death process: $\beta=0.4$, $\gamma=1$, $M=100$, simulation time $T_{\text{end}}=80$, resampling interval $T_{\text{step}}=4$, burn-in $T_b=40$, sampling delay $T_d=2$.} 
	\label{fig: lin bdp}
\end{figure}


\section{Regional resampling}\label{regional}
\subsection{Idea and algorithm}
One difficulty with SMC sampling on a reducible state space is that once all the particles have left a transient communicating class there is no mechanism for particles to return there. Since the support of the initial distribution determines which of the LCDs is being estimated, this changes the LCD that the SMC sampler is converging towards so that it is no longer the required target distribution. To address this weakness, we propose \emph{regional resampling} in which the state space is partitioned into
regions and the number of particles available to explore each region is stipulated in advance. At each resampling timepoint, particles are removed from regions with too many particles and reallocated to regions with too few. 
We describe regional resampling in the context of LCDs, we anticipate that it will have wider applications in more general SMC schemes.
\edit{Intuitively, this can be seen as just a redistribution of a resampling step using existing methods between different sections, so the resampling holds as expected locally, and the preservation of region weight allows it to hold globally.}

Let $X = \{X(t):t\geq 0\}$ be an absorbing Markov process on a state space \mbox{$S \cup \{0 \}$}, with absorbing state $0$. We partition the transient states into $L$ \edit{further regions:} $S = \bigcup_{l=1}^L S_l$. For each region $l=1,\dots,L$ we specify $N_1, N_2, \dots, N_L > 0$ to be the desired number of particles in each region, such that $\sum_{l=1}^L N_l = M$. 

Each resampling step begins with a set of weighted particles $\{(X_i, W_i): i=1, \dots, M\}$. Any absorbed particles have weight zero. Let $M_l(t)$ be the number of particles in region $l$ at  resampling time $t$ and define $W(l)$ to be the total weight in region $l$; $W(l)=\sum_{i=1}^MW_i\1_{\{X_i\in S_l\}}$. 

In regional resampling, we resample particles in each region separately. In region $l$ we draw $N_l$ particles from the $M_l(t)$ existing particles in that region using any resampling algorithm (such as combine-split, particle refilling or an existing algorithm as described in \ref{subsec: resample} with $M'=N_l$). 
The weights of these new particles are then renormalized so that the total weight in region $l$ remains equal to $W(l)$. For example, after multinomial resampling within region $S_l$, the renormalized weights of each particle would be $W(l)/N_l$. \edit{It should be noted that between two resampling steps it is possible that a region may run out of particles and so resampling may fail. This problem will be tackled using stopping times in Section \ref{stopping}.}

In sampling from LCDs we expect that $N_l\geq M_l(t)$ in most cases, because some of the particles will have been absorbed. However it is possible that in some regions $N_l<M_l(t)$ and so if we are applying combine-split or particle refilling it may not be possible to keep all of the locations. 
It is therefore necessary to apply an alternative resampling step such as multinomial resampling in these instances. 




It should be noted that one does not need to choose $S_l$ to be a single transient communicating class. This method can be extended to any \edit{complete partition $\{S_l\}_{l=1}^L$ of $S$.}
\edit{As for combine-split resampling, this inherits many convergence properties of SMC samplers, and using a deterministic resampling timepoint sequence allows for very effective parallelization of this SMC sampler.}
\begin{prop}\label{regionalProper}
Given a properly weighted sample $\{(X_j,w_j):1\leq j \leq M\}$ then regional resampling, using a resampling method that produces properly weighted samples within each region, produces properly weighted samples.
\end{prop}
\begin{proof}
Let $(X_j',w_j')$ be the location and weight of particle $j$ after resampling. Then since resampled particles are properly weighted within each region, $\EE[h(X_j')w_j']=\sum\limits_{l=1}^L\EE[h(X_j')w_j'\1_{\{X_j'\in S_l\}}]=\sum\limits_{l=1}^Lc_l\EE_\pi[h(X)\1_{\{X\in S_l\}}]=c\EE_\pi[h(X)]$, where $c=\sum_{l=1}^L c_l$ for all $j$. Hence regional resampling produces properly weighted particles. $\qedsymbol$
\end{proof}
\subsection{Proof of convergence for a simple example}
In general it is extremely challenging to prove convergence of such SMC sampler methods with resampling, except in very specific circumstances. We present here a simplified model to \edit{prove} that the algorithm 
\edit{converges} in some situations. \edit{This stylised example is chosen because more general reducible processes could be used and achieve the same result but the result itself would be more complex without any additional intuitive illumination. The general SMC method is shown to converge, under multinomial resampling, and \edit{without irreducibility conditions as is also the case in Proposition 9.4.1 of \cite{delmoral2004}}, but our new resampling method extends this, since we also include the refilling methods defined in Section \ref{subsec: refilling}.}

Consider the pure death process on $\{0,1,2\}$. Transitions from state 2 to state 1 occur with rate $\delta$ and transitions from state 1 to 0 (absorption) occur with rate $1$. When $0 < \delta < 1$, the $2$-LCD is given by $(\delta, 1-\delta)$; when $\delta\geq 1$ the 2-LCD is simply $(1,0)$. We wish to prove that the SMC sampler with regional resampling converges. 
Choose two regions given by $S_1 = \{1\}$, $S_2 = \{2\}$. We stipulate $N_1, N_2$ with $N_1 + N_2 = M$ to be the desired populations of each region. In what follows, we let $w_l(t)$ be the unnormalized weights on the particles in state $l$ at time $t$.
We want to show that the normalized weights $(W_1(t), W_2(t))$ converge to some distribution for any choice of $(W_1(0)$, $W_2(0))$, if we take the limit $t\rightarrow \infty$ at the sequence of resampling times $(t_n)_{n\geq 1}$.

We look at the simplest case where combine-split resampling happens within each region after every event in the simulation. In this case, we can see the unnormalized weights as moving at points of a Poisson process with rate $N_1+\delta N_2$, with the jump chain $(w_1, w_2)(n) = (w_1(t_n), w_2(t_n))$ moving like
\begin{align*}
  (w_1, w_2)(n+1) =  \begin{cases}
                                 \left( w_1(n)\frac{N_1 - 1}{N_1}, w_2(n) \right) & \text{w.p.} \quad \frac{N_1}{N_1+\delta N_2}, \\
                                 \left(w_1(n) + \frac{w_2(n)}{N_2}, w_2(n)\frac{N_2 - 1}{N_2} \right) & \text{w.p.} \quad  \frac{\delta N_2}{N_1 + \delta N_2}.
                               \end{cases}
\end{align*}
\edit{where ``w.p.'' means ``with probability''.}
Each jump immediately triggers a combine-split resampling within each region. In this model the regions consist of just one state, and so the combine-split resampling simply spreads the weight uniformly amongst the particles in that region.

If we define $X(n) := W_1(t_n)$, this Markov chain evolves according to:
\begin{align}\label{eqn: X}
  X(n+1) = 
  \begin{cases}
		\frac{X(n)( N_1 - 1)}{N_1 - X(n)} & \text{w.p.} \quad  \frac{N_1}{N_1 + \delta N_2} \\
		X(n) + \frac{1-X(n)}{N_2} & \text{w.p.} \quad \frac{\delta N_2}{N_1+\delta N_2}
	\end{cases}
\end{align}
Since we know that $W_2(t_n) = 1 - X(n)$, if we can show convergence in Wasserstein Distance of $X(n)$ then we have convergence in Wasserstein distance of both weights.

\begin{def}\label{def: dWass}
Let $E$ be a Polish space, and let $d: E \times E \rightarrow [0,1]$ be a distance-like function (symmetric, lower-semi-continuous and $d(x,y)=0 \Leftrightarrow x=y$). Then the \emph{Wasserstein-d distance} between two probability measures $\mu$, $\nu$ on $E$ is given by
		\[
		  \cW_d(\mu, \nu) = \inf_\pi \int_{E \times E} d(u,v) \pi(du,dv)
		\]
where $\pi$ runs over all probability measures on $E \times E$ which have marginals $\mu$, $\nu$.
\end{def}

\begin{thm}\label{thm: Wass}
  The distribution of $X(n)$ on $(0, 1)$ defined in \eqref{eqn: X} converges in Wasserstein-$d$ distance with 
	\[
d(x,y) = \min \left\{ 1, \left| (1-x)^{-1} - (1-y)^{-1} \right| \right\}
\]
to some stationary distribution $\pi$, for all $\delta \in (0,1)$, $N_1 \geq 2$, $N_2 \geq \max(5, \frac{1}{1-\delta})$.
\end{thm}

We prove this be applying the following theorem taken from \cite{hairer2014}.
\begin{thm}\label{thm: hairer}
  Let $P$ be a Markov kernel over a Polish space $E$ and assume that:
  \begin{enumerate}
    \item $P$ has a Lyapunov function $V: E \rightarrow \RR$ such that there exists $\lambda \in [0, 1)$ and $K > 0$ such that for all $x \in E$,
    \[
      PV(x) := \int_E V(u) P(x,du) < \lambda V(x) + K
    \]
		
		\item P is $d$-contracting for some distance-like function $d: E\times E \rightarrow [0,1]$ ($d$ is symmetric, lower-semi-continuous and $d(x,y)=0 \Leftrightarrow x=y$), so that there exists $c \in (0,1)$ such that for every $x,y \in E$ where $d(x,y)<1$ we have
		\[
		  \cW_d(P(x,\cdot), P(y, \cdot)) < c \, d(x,y).
		\]
		\item The set $S = \{x: V(x) < 4K\}$ is $d$-small, so that there exists $s \in (0,1)$ such that for all $x,y \in S$
		\[
		  \cW_d(P(x,\cdot), P(y, \cdot)) \leq s.
		\]
  \end{enumerate}
  Then there exists $n \in \NN$ such that for any two probability measures $\mu$, $\nu$ on $E$ we have
  \[
    \cW_{\tilde{d}}(\mu P^n, \nu P^n)  \leq \cW_{\tilde{d}}(\mu, \nu)
  \]
	where $\tilde{d}(x,y) = ( d(x,y)(1 + V(x) + V(y)) )^{1/2}$, and $n$ is increasing in $\lambda, K, c, s$. Hence there is at most one invariant measure. 
	
	Moreover, if the following hold:
	\begin{enumerate}
	\setcounter{enumi}{3}
	\item There exists a complete metric $d_0$ on $E$ such that $d_0 \leq \sqrt{d}$
	\item $P$ is Feller on $E$, which holds precisely when for any continuous function $f$ on $E$, $\int_E f(y) P(x,dy)$ is continuous for every $x \in E$,
	\end{enumerate}
	then there exists a unique invariant measure $\mu$ for $P$. [Hairer, Stuart, Vollmer]
\end{thm}

The statement follows by the application of the following lemmas, proofs of which can be found in Chapter 5 of  \cite{griffin2016}.

\begin{lem}\label{lem: W1}
	The function $V(x) = x(1-x)^{-1}$ is Lyapunov for $P$ as above with
	\begin{align*}
	\lambda = 1 - \frac{N_2(1 - \delta)  - 1}{(N_2 -1)(N_1 + \delta N_2)} & &
	K = \frac{\delta N_2}{(N_2 - 1)(N_1 + \delta N_2)}
	\end{align*}
\end{lem}
\begin{lem}\label{lem: W2}
	$P$ is $d$-contracting for distance like function 
	\[
	d(x,y) = \min \left\{ 1, \left| (1-x)^{-1} - (1-y)^{-1} \right| \right\}
	\]
\end{lem}
\begin{lem}\label{lem: W3}
	The set $S = \{x: V(x) < 4K \}$ is $d$-small with $V(x)$ and $K$ as defined in Lemma \ref{lem: W1}, and $d(x,y)$ as in Lemma \ref{lem: W2}.
\end{lem}

\begin{proof}[Proof of Theorem \ref{thm: Wass}]
Lemmas \ref{lem: W1}, \ref{lem: W2} and \ref{lem: W3} prove that conditions 1,2 and 3 of Theorem \ref{thm: hairer} hold, which gives us that there exists at most one invariant measure.

To prove the existence of the invariant measure we need to satisfy the additional conditions 4 and 5. Since $P(x, \cdot)$ is a finite sum of atomic measures for every $x$ in $E$, we get that $P$ is Feller on $E$, satisfying condition 5. For condition 4 we look for a complete metric $d_0 \leq \sqrt{d}$. Indeed, since we can consider the process $X$ to be defined on $[0,1]$, we do this by extending the distance-like function $d$ to include 0 and 1:
\begin{align*}
d(x,1) = 1& &d(1,1)=0& & d(x,0)=\min \left\{ \left| \frac{1}{1-x} - 1 \right| , 1\right\}.
\end{align*}
Since for all $x,y\in [0,1]$ we have that $d(x,y)\in [0,1]$, we have that
$\sqrt{d(x,y)} \geq d(x,y) = \frac{|x-y|}{(1-x)(1-y)} \geq |x-y|.$
Setting $d_0$ to be the Euclidean metric, which is complete on $[0,1]$ gives us the required condition. Therefore condition 4 of Theorem \ref{thm: hairer} is satisfied and hence there exists a unique invariant measure for the process $X$. $\qedsymbol$
\end{proof}

Now we know that this Markov chain converges in a Wasserstein distance to some limiting distribution on $(0, 1)$, we would hope that the mean of the limiting distribution is close to the true value we want: $\delta$. For linear systems this is simple to compute, but this is not the case for our process. However, one can see in Figure \ref{fig: conv} that for even modest $N_1$, $N_2$ the process above does converge to a value close to the true $\delta$. \edit{We note that this could be done on a more general reducible process on $\{0, \dots, L\}$ but would involve multivariate processes $(W_1, W_2, \dots)$ which would make the proof more complex but give the same result.}

\begin{figure}[h!]
	\centering
	\begin{subfigure}[b]{0.45\textwidth}
		\includegraphics[width=\textwidth]{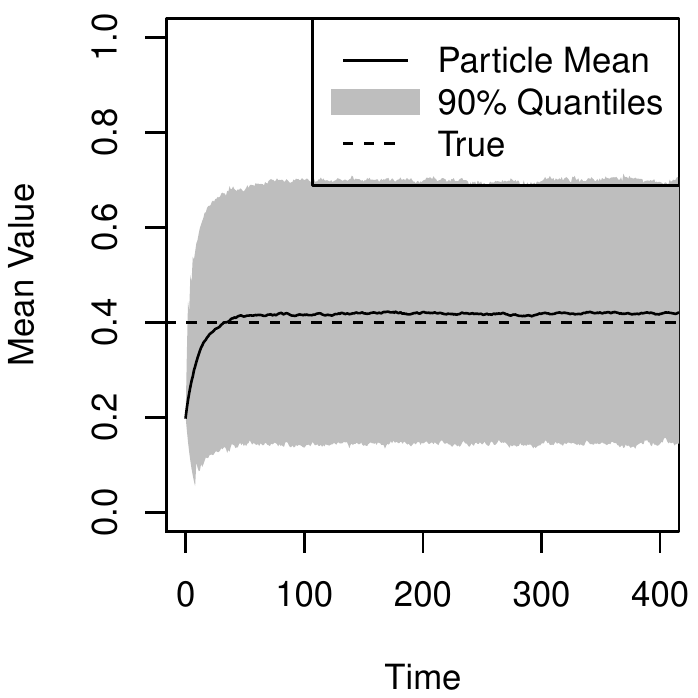}
		\caption{$N_1 = 6$, $N_2 = 6$}
		\label{fig: conv a}
	\end{subfigure}\quad
	\begin{subfigure}[b]{0.45\textwidth}
		\includegraphics[width=\textwidth]{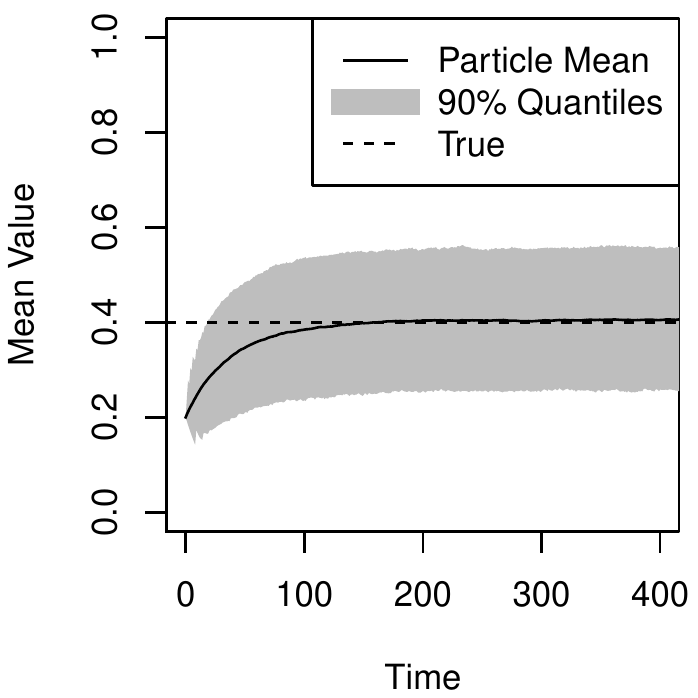}
		\caption{$N_1 = 20$, $N_2=20$}
		\label{fig: conv b}
	\end{subfigure}
	\caption{Mean convergence of process as $N_1, N_2$ change: $\delta = 0.4$, $M=5000$ particles.}
	\label{fig: conv}
\end{figure}
	
\subsection{Example: Pure death process}\label{puredeath}
To further demonstrate regional resampling we consider again the pure death process on a finite state space: $S = \{1,2,\dots,K\}$ with absorbing state $0$, described in Section \ref{subsubsec: motive}. 
The LCDs of the finite pure death process are classified in Theorem \ref{thm: d LCD}, which is proven in \cite{griffin2016}, and is an extension of a theorem in \cite{vandoorn2008}.
\begin{thm}\label{thm: d LCD}
	Let $(X_t)_{t \geq 0}$ be a pure death process on $\{0, 1, \dots, L\}$ with death rates $\{ \delta_i: 1\leq i \leq L\}$. Then for each $i \in \{1, \dots, L \}$ there exists a unique $i$-LCD which gives mass to states $\{1, \dots, L(i)\}$ where $L(i)=\max\{j:1\leq j \leq i, \delta_k\geq\delta_j \edit{\forall} k=1,\dots,j\}$. 
\end{thm}
Theorem \ref{thm: d LCD} can be understood in terms of bottlenecks in the flow towards zero: the process conditioned on non-extinction will sit in states up to the narrowest bottleneck below the starting position.
The pure death process has a reducible state space, in fact no state can be returned to once it has been left. Such reducible state spaces present problems when simulating the LCD using SMC samplers because if at any point in time all of the particles have left a transient communicating class then this class can never be returned to. Regional resampling provides a way to ensure that the number of particles in each region is maintained.

Consider a toy example of a pure death process with transient states $S = \{1, \dots, 5\}$ and death rates given by $\delta = (\delta_1, \dots, \delta_5) = (3,2,3,1,3)$. By solving the left eigenvalue equations we see that this process has 5-LCD given by $\bu=(1/3,1/3,1/9,2/9,0)$.
\begin{figure}[h!]
	\centering
	\begin{subfigure}[b]{0.45\textwidth}
		\includegraphics[width=\textwidth]{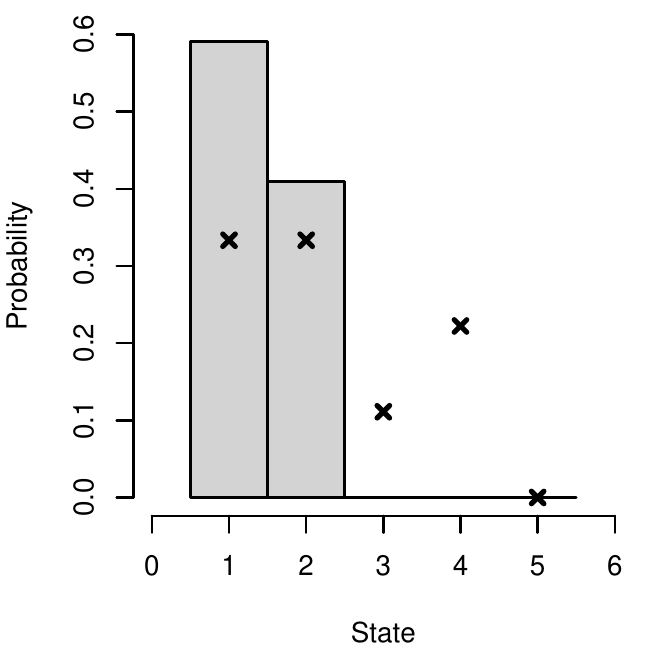}
		\caption{Simulation from 5-LCD with particle refilling}
		\label{fig: refill d}
	\end{subfigure}\quad
	\begin{subfigure}[b]{0.45\textwidth}
		\includegraphics[width=\textwidth]{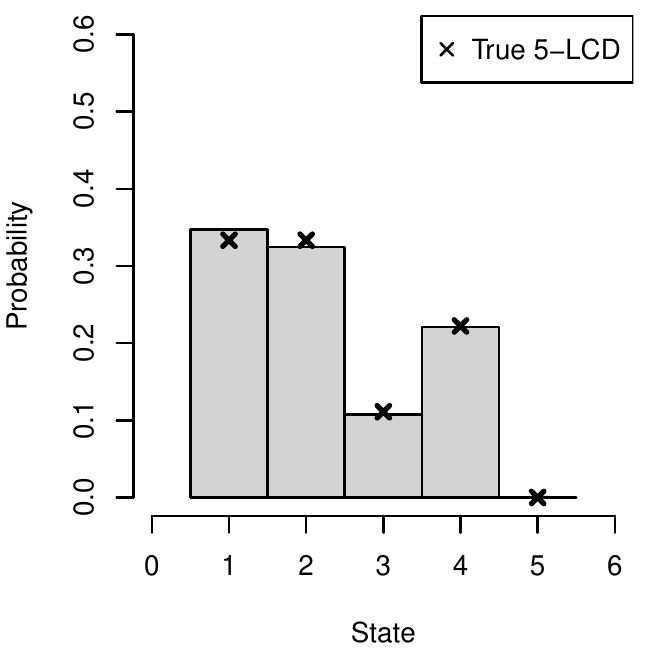}
		\caption{Simulation of 5-LCD with 2-region resampling}
		\label{fig: split d}
	\end{subfigure}
		\centering
	\begin{subfigure}[b]{0.45\textwidth}
		\includegraphics[width=\textwidth]{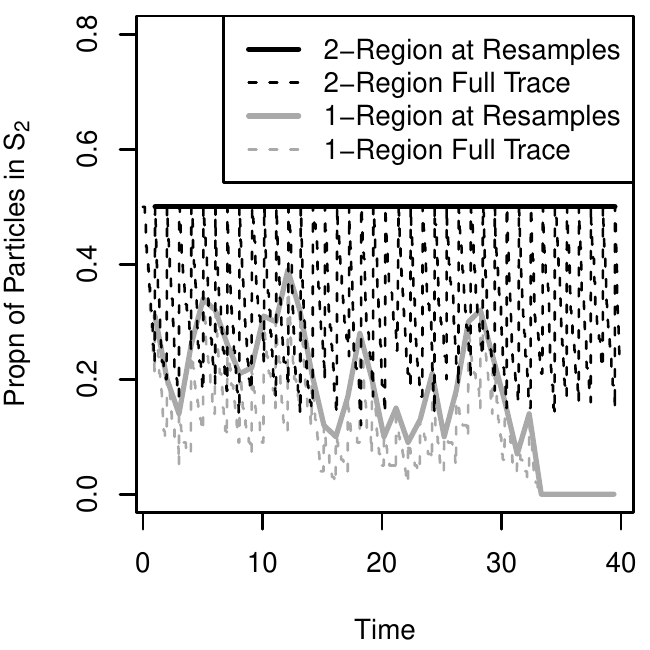}
		\caption{Proportion of particles in region $S_2$}
		\label{fig: props2}
	\end{subfigure}\quad
	\begin{subfigure}[b]{0.45\textwidth}
		\includegraphics[width=\textwidth]{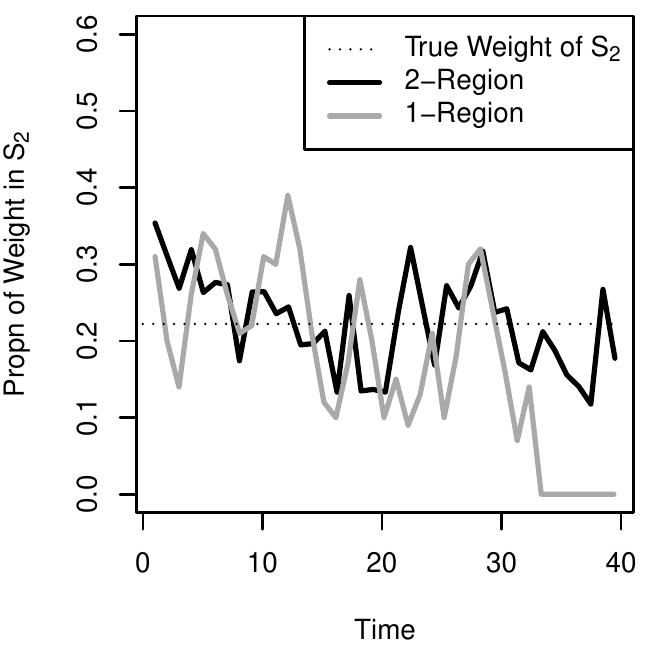}
		\caption{Proportion of total weight in $S_2$}
		\label{fig: ess12}
	\end{subfigure}
	\caption{Region Resampling for the 5-state Death Process with \\ $\delta = (3,2,3,1,3)$. $M=100$, $T_{\text{end}} = 40$, $T_{\text{step}} = 1$, 
	\\ For (a),(b) the simulation uses burn-in $T_b=20$, sampling delay $T_d=2$.}
	\label{fig: d}
\end{figure}

Figure \ref{fig: d} shows the simulated 5-LCDs for this process using particle refilling (\ref{fig: refill d}) and regional resampling with 2 regions (\ref{fig: split d}). The regional resampling was performed over the regions $S_1=\{1,2\}$ and $S_2=\{3,4,5\}$. Under particle refilling at some point all of the particles leave $S_2$ and so the SMC converges towards the 2-LCD instead of the 5-LCD. Although 5 regions could have been specified to reflect the 5 transient communicating classes, in this example 2 regions were sufficient for the SMC sampler with regional resampling to converge to the correct distribution.

Figures \ref{fig: props2} and \ref{fig: ess12} demonstrate how the proportion of particles and proportion of weights differ under particle refilling and regional resampling schemes. Although the number of particles following a resample in region $S_2$ is fixed for regional resampling, the proportion of weight in $S_2$ is not fixed, gradually converging to some value. Despite depletion between resampling times, at no point does the number of particles in $S_2$ reach zero, which ensures that regional resampling converges to the correct LCD. 
	

\section{Dynamic resampling schemes}\label{stopping}

When using an SMC sampler to draw samples from a LCD on a reducible state space, one potential problem is that all of the particles can leave a transient communicating class between resampling times. Regional resampling alleviates this to a certain extent, by ensuring that at each resampling time there are a prespecified number of particles in each region, however there still remains a positive probability that all of these particles could leave the region they started from if the time between resampling events is fixed. 
Furthermore, as the duration of the simulation increases so will the number of resampling events, and hence the probability of a region being left without any particles converges to 1 -- the SMC sampler will almost surely fail in a finite time. Clearly this property is extremely undesirable. However, if the resampling times are not fixed in advance, but instead are allowed to depend on the progress of the particles dynamically via \edit{some stopping times} \cite{liu1995, liu1998}, 
then this property can be avoided. Under typical dynamic resampling schemes, resampling is activated when particle diversity drops below some threshold, where diversity is measured by the coefficient of variation of the importance weights. However, in light of the discussion surrounding regional resampling, here we introduce a sequence of stopping times which activate resampling based on the diversity of particle locations. Other sequences of stopping times may be possible, along the lines of \emph{multi-level SMC} / \emph{stopping-time resampling} \cite{ delmoral2004, jenkins2012}. In those algorithms, particles evolve independently between resampling events and stopping times are measurable with respect to the trajectory of a single particle.

Suppose that we wish to use an SMC sampler with regional resampling to simulate from the LCD of a process $(X(t))_{t \geq 0}$ on a state space $S \cup \{0\}$. As in Section \ref{regional}, we partition the transient states into $L$ regions, where $S = \bigcup_{l=1}^L S_l$. Let $M_l(t)$ be the number of particles in region $S_l$ at time $t$ and choose $(N_1, \dots, N_L)$ to be the desired number of particles in each region after a resampling event. We wish to initiate a resampling event whenever the number of particles in region $l$ drops below $\lambda N_l$, where $\lambda\in (0,1)$ is some tuning parameter that controls how often resampling occurs.

More formally, we define the following sequence of stopping times. Let $T_0=0$, then for $k\in\mathbb{N}$, let $T_k = \min \{ T_k^{(1)}, \dots, T_k^{(L)}, (T_{k-1} + T_{\max}) \}$
where the local stopping time for region $S_l$ is given by $T_k^{(l)} = \inf \{ t > T_{k-1} : M_l(t) \leq \lambda N_l \}$.
The parameter $T_{\text{max}}\in(0,\infty]$ controls the maximum time between resampling events. At each stopping time $T_1,T_2,\dots$, a regional resampling event is triggered in all $L$ regions. Assuming that $N_l\geq 2$ for all $l$, then there must be at least one particle to resample from within each region at the time that a resampling event is triggered.

\subsection{Resampling timepoint optimisation}
The main advantage to the dynamic resampling is that the SMC sampler will definitely terminate successfully at the prescribed endpoint, whereas with all the deterministic timepoint methods, there is a non-zero probability that during the simulation, a resampling step will be unable to take place. The second advantage is also that one can tune $\lambda$, which dictates how low a region's population must get to trigger a resampling step. In doing so, one also tunes the number of resampling steps, and the speed of convergence in mean to the true LCD. 

\subsection{Example: Transient immunity process.}\label{subsec: ti}

The transient immunity process is an extension of the linear birth-death process in which individuals persist after `death' for a random time in an additional non-reproductive state. The process $X(t) = (I(t), R(t))_{t \geq 0}\in\mathbb{N}_0\times\mathbb{N}_0$ evolves as follows. New infectious individuals are created at the points of an inhomogeneous Poisson Process with rate $\beta I(t)$, at which point $I(t)\mapsto I(t)+1$. Each infection lasts an exponentially-distributed time with rate $\gamma$, and so recoveries occur at the points of a Poisson process with rate $\gamma I(t)$, at which point $(I(t),R(t))\mapsto (I(t)-1,R(t)+1)$. After recovery, individuals experience a period of immunity which lasts an exponentially distributed time with rate $\delta$. These loss of immunity events occur at the points of a Poisson process with rate $\delta R(t)$, at which $R(t)\mapsto R(t)-1$.
 
The state space of the transient immunity process has countably many communicating classes: $\{ (0,r) \}$ for each $r > 0$ plus $S_2 = \{ (i,r) : i > 0, r \geq 0 \}$.  We consider the absorbing state to be $(0,0)$ when both the infection and the immunity have left the population. Despite its relatively simple linear structure, a full characterisation of the LCDs is not available.

Work from \cite{vandoorn2013} can be extended using a coupling with the linear birth-death process to show that QSDs exist for all choices of $\beta$, $\gamma$, $\delta > 0$ such that $\gamma - \beta > 0$. Furthermore one can show that uniqueness holds for $\bu$ the $\mathbf{v}$-LCD for all $\mathbf{v}$ with finite support contained in $S_2$. 
The characterisation of $\bu$ depends on whether $\gamma - \beta > \delta$ or $\gamma - \beta < \delta$. If $\gamma - \beta > \delta$ then $\bu$ is a unit mass on the point $(I,R) = (0,1)$. If $\gamma - \beta < \delta$, one can show that $\bu$ gives mass to all states in $S$.

We first show how 2-region resampling with dynamic resampling works for the transient immunity process to give an idea of the character of the $(1,0)$-LCD, making use the two regions $S_1 = \bigcup_{r=0}^\infty \{(0,r)\}$ and $S_2 = \{ (i,r) : i > 0, r \geq 0 \}$, so $S = S_1 \cup S_2$. In Figure \ref{fig: rr low ti}, we see that under 2-region resampling, despite having $40\%$ of the particles in $S_2$ we see a negligible contribution from them. On the other hand, in Figure \ref{fig: rr high ti}, where the true $(1,0)$-LCD gives weight to all states in the state space, the simulation gives a good visualisation of the $(1,0)$-LCD.

\begin{figure}[h!]
	\centering
	\begin{subfigure}[b]{0.49\textwidth}
		\includegraphics[width=\textwidth]{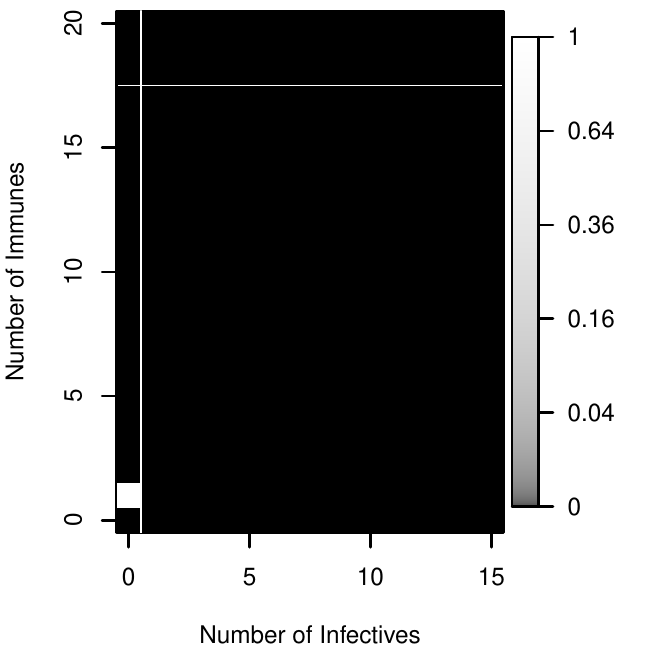}
		\caption{$\gamma-\beta=0.5$, $\delta=0.4$}
		\label{fig: rr high ti}
	\end{subfigure}
	\begin{subfigure}[b]{0.49\textwidth}
		\includegraphics[width=\textwidth]{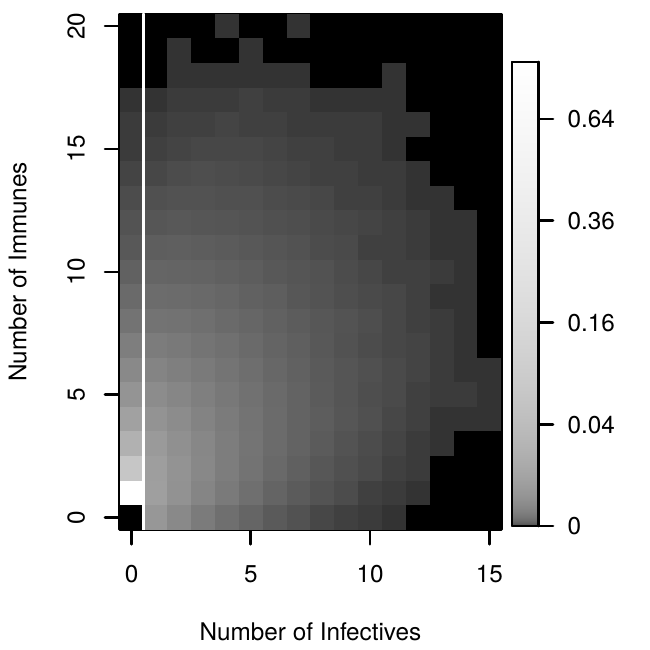}
		\caption{$\gamma-\beta=0.5$, $\delta=0.6$}
		\label{fig: rr low ti}
	\end{subfigure}
	\caption{Comparison of simulated LCDs from 2-region resampling: $M=6000, (N_1, N_2)=(4000,2000)$, $T_{\text{end}} = 450$, $T_{\text{max}} = 12$, $T_b=40$, $T_d=1$.}
	\label{fig: rr ti}
\end{figure}
	
\begin{figure}[h!]
	\begin{subfigure}[b]{0.45\textwidth}
		\includegraphics[width=\textwidth]{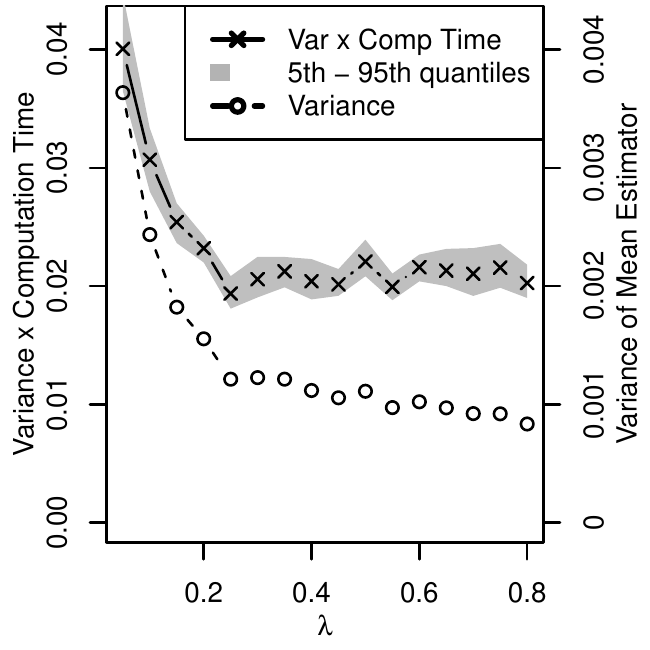}
		\caption{Performance of estimator under varying $\lambda$: $T_{\max} = 10$, computed from 2000 runs.}
		\label{fig: varX}
	\end{subfigure}
	\begin{subfigure}[b]{0.49\textwidth}
		\begin{center}
			\begin{tabular}{cccccc}\hline
				$\lambda$ & $T_{\max}=1$ & $2$ & $3$ & $4$ & $5$\\\hline
				0 & 200 & 100 & 67 & 50 & 40\\
				0.2 & 200 & 100 & 68 & 61 & 63\\
				0.4 & 200 & 107 & 105 & 105 & 105\\
				0.6 & 204 & 182 & 184 & 183 & 180\\
				0.8 & 395 & 395 & 394 & 397 & 397\\\hline
				\\
				\\
			\end{tabular}
		\end{center}
		\caption{Number of resampling steps as $\lambda$ and $T_{\max}$ change}
		\label{tab: lambda}
	\end{subfigure}
\caption{Comparison of performance under varying $\lambda$ and $T_{\max}$:\\$\beta-\gamma=0.5$, $\delta=0.6$, $T_{\text{end}} = 80$, $T_b=0$, $T_d=1$, $M=500$, $(N_1, N_2)=(300,200)$ \\
Computation performed on a single core of an Intel Core i5-vPro processor and took 13 seconds for $\lambda = 0.6$, $T_{\max} = 3$}
	\label{fig: var lam}
\end{figure}

To illustrate how the choice of $\lambda$ and $T_{\max}$ affect performance of the model, we consider how consistently such methods estimate the mean of the QSD for the transient immunity process seen previously. Under all choices of $\lambda$ and $T_{\max}$ we obtained similar estimates of the mean, but the variability of such estimates is of interest. We first fix $T_{\max} = 10$. 
By repeating our simulations we measure the Monte Carlo variance of the estimator of the mean number of infectives, and the mean computation time to run a single simulation run of $M=500$ particles; this computation time is only meant to be illustrative, and depends on the hardware and software used in general. From Figure \ref{fig: varX} we see that the variance does indeed decrease as $\lambda$ increases, but the improvement is noticeably smaller beyond $\lambda = 0.25$. 
If we take the product of the variance and mean computation time, then we see it is indeed optimal at this value of $\lambda$. For different uses, this optimal value of $\lambda$ may vary, and if computation time is not an issue, one can instead choose as large a $\lambda$ as possible.

If we furthermore vary $T_{\max}$ we see in Figure \ref{tab: lambda} that for different values of $\lambda$, and for sufficiently large $T_{\max}$, the precise choice of $T_{\max}$ is unimportant: there seems to develop a natural frequency of resampling steps. As such, it is felt that one should include a $T_{\max}$ which is slightly larger than the natural period of resampling under $\lambda$, since this $T_{\max}$ then still performs the job of avoiding particle weight degeneracy as in standard deterministic timepoint resampling schemes, but does not trigger unnecessarily frequent resampling steps. 
For $\lambda=0 $ and $T_{\max} \geq 3$, the probability of a successful simulation run was low due to the high rate of particle extinction, and so repeated simulations were required. 

\edit{
 \section{Estimating the Decay Parameter / Rate of Extinction}\label{application}
 Using the tools we have developed, we apply them to tackling an important problem in the field of QSDs, one of determining the eigenvalue associated to a given QSD or LCD $\bu$: the value of $\alpha$ for which $-\alpha \bu^{T} = \bu^{T}Q$. For irreducible processes, the maximal $\alpha$ is known as the \emph{decay parameter} \cite{vandoorn2014}.}
 
%
 To compute the decay parameter when the state space is finite may require finding the eigenvalues of a large $Q$ matrix, which is computationally intensive and numerically unstable. In the countable case, the required knowledge of $P[X(t) = i | X(0) = j]$, is frequently intractable. 
To this end, we apply the SMC sampler to obtain estimators for $\alpha$ using the following result from \cite{vandoorn2013}.
 \begin{prop}
   Let $X$ be an absorbing process on $S\cup\{0\}$ with absorbing state $0$, and let $\bu$ be a QSD for $X$. If $\bu$ is $\alpha$-invariant for $Q$, that is, \edit{$-\alpha \bu^{T} = \bu^{T}Q$},
 where $Q$ is the generator matrix $\wtilde{Q}$ restricted to $S$, then we have that $\alpha = \sum_{s \in S} u_s \tilde{q}_{s0}$, where $\tilde{q}_{s0}$ is the rate of moving from $s$ to $0$.
 \end{prop}
 
 In models in which there is only one state from which the process can reach $0$ (for example population processes), only a single term contributes to the sum for $\alpha$. For example, in the pure death process we have $\alpha = u_1 \delta_1$. If one can draw iid $X_j \sim u$ for $j=1, \dots, M$ then 
 \[
   \widehat{\alpha} := \sum_{s \in S} q_{s0} \sum_{j=1}^M M^{-1} \ind_{\{X_j = s\}}
 \]
 is an unbiased estimator for $\alpha$. This estimator can therefore be implemented as follows. Following a burn-in period $T_b$, at each sampling time $t$ we record $\widehat{\alpha}$ by replacing $M^{-1}$ by normalized weights $W_j(t)$.
%
 \subsection{Example: Transient immunity process}
 
 We estimate the decay parameter $\alpha$ for the transient immunity process defined in Section \ref{subsec: ti}. Since the only exit route is via $(I,R) = (0,1)$ our estimator \edit{$\widehat{\alpha}$} reduces to $\widehat{\alpha} = \delta \sum_{j=1}^M W_j \ind_{\{X^j = (0,1)\}}$. \edit{The eigenvalue $\alpha$ associated to the LCD starting from the state $(1,0)$ is given by $\alpha = \min(\delta, \gamma - \beta)$ \cite{griffin2016}.}
 
 Figure \ref{fig: alphati} uses dynamic regional resampling with combine-split particle reallocation to estimate $\alpha$. The estimates are accurate in the two regimes of the process, where $\delta$ is larger (or smaller) than $\gamma - \beta$. However, stochastic effects cause the decay parameter to be underestimated close to the critical value $\delta=\gamma-\beta$.
 \begin{figure}[h!]
  	\centering
  	\includegraphics[width=0.6\textwidth]{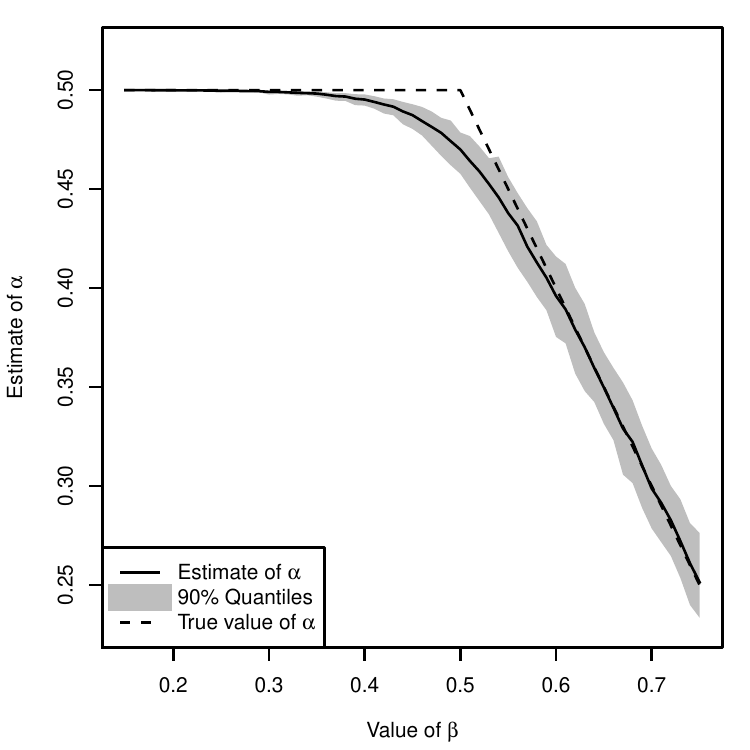}
  	\caption{Estimate of $\alpha$ as $\beta$ changes. $\gamma = 1$, $\delta=0.5$, $T_{\text{end}}=60$, $T_{\max}=5$, $T_b = 20$, $T_d= 1$, $M=400$, $N=(200,200)$}
  	\label{fig: alphati}
 \end{figure}

\section{Conclusion}
In this paper we have developed a toolbox of techniques for simulating from LCDs on reducible state spaces using SMC samplers. We have shown how combine-split resampling can improve particle diversity when the state space is discrete. We introduced regional resampling to allow more control of the distribution of the particles in the SMC sampler. Although our focus was on simulating LCDs on reducible state spaces, we anticipate that regional resampling could prove useful for SMC samplers designed for other purposes. We also demonstrated that by adopting a dynamic resampling scheme based on the number of particles in each region, we could prevent failures in the SMC sampler that were otherwise certain to occur in finite time.

Since there is always a finite number of particles in the SMC sampler, it is not possible to use these techniques to sample from the `high energy' QSDs that sometimes exist for processes on infinite state spaces, as these distributions have infinite mean. Designing a mechanism to sample from such LCDs remains an open problem. To further this work, we would like to know whether there is a more systematic approach to the selection of $\lambda$ and $(N_1, \dots, N_L)$ to best approximate the true LCD. \edit{Moreover, in the case of complex models where regions of interest are not known, dynamic regional resampling where regions update according to certain statistics could be investigated.}

\ack
Adam Griffin's research is funded by EPSRC: Grant Number EP/HO23364/1.

\bibliographystyle{apt}
\bibliography{biblio3}

\begin{thebibliography}{10}

\bibitem{blanchet2013}
{\sc Blanchet, J., Glynn, P. and Zheng, S.} (2016).
\newblock Analysis of a stochastic approximation algorithm for computing
  quasi-stationary distributions.
\newblock {\em Adv. Appl. Prob.\/} {\bf 48,} 792--811.

\bibitem{clancy2003}
{\sc Clancy, D. and Pollett, P.} (2003).
\newblock A note on quasi-stationary distributions of birth-death processes and
  the {SIS} logistic epidemic.
\newblock {\em J. App. Prob.\/} {\bf 40,} 821--825.

\bibitem{darroch1967}
{\sc Darroch, J. and Seneta, E.} (1967).
\newblock On quasi-stationary distributions in absorbing continuous-time finite
  {M}arkov chains.
\newblock {\em J. App. Prob.\/} {\bf 4,} 192--196.

\bibitem{delmoral2004}
{\sc {Del Moral}, P.} (2004).
\newblock {\em Feynman-Kac Formulae: Genealogical and Interacting Particle
  Systems with Applications}.
\newblock Springer-Verlag, New York.

\bibitem{delmoral2006}
{\sc Del~Moral, P., Doucet, A. and Jasra, A.} (2006).
\newblock Sequential {M}onte {C}arlo samplers.
\newblock {\em J. Roy. Stat. Soc. B (Met.)\/} {\bf 68,} 411--436.

\bibitem{douc2008}
{\sc Douc, R. and Moulines, E.} (2008).
\newblock Limit theorems for weighted samples with applications to sequential
  monte carlo methods.
\newblock {\em The Ann. of Stats\/} 2344--2376.

\bibitem{doucet2001}
{\sc Doucet, A., De~Freitas, N. and Gordon, N.} (2001).
\newblock {\em Sequential {M}onte {C}arlo methods in practice}.
\newblock Springer.

\bibitem{etheridge2011}
{\sc Etheridge, A.} (2011).
\newblock {\em Some Mathematical Models from Population Genetics: {\'E}cole
  D'{\'E}t{\'e} de Probabilit{\'e}s de Saint-Flour XXXIX-2009} vol.~39.
\newblock Springer Science \& Business Media.

\bibitem{griffin2016}
{\sc Griffin, A.} (2016).
\newblock Quasi-stationary distributions for epidemic models: Simulation and
  characterisation.
\newblock {\em PhD thesis}.

\bibitem{groisman2012}
{\sc Groisman, P. and Jonckheere, M.} (2013).
\newblock Simulation of quasi-stationary distributions on countable spaces.
\newblock {\em Markov Proc. \& Rel. Fields\/} {\bf 19,} 521--542.

\bibitem{hairer2014}
{\sc Hairer, M., Stuart, A. and Vollmer, S.} (2014).
\newblock Spectral gaps for a {M}etropolis--{H}astings algorithm in infinite
  dimensions.
\newblock {\em Ann. App. Prob.\/} {\bf 24,} 2455--2490.

\bibitem{jenkins2012}
{\sc Jenkins, P.} (2012).
\newblock Stopping-time resampling and population genetic inference under
  coalescent models.
\newblock {\em Stat. App. in gen. and Mol. Bio.\/} {\bf 11,} Article 9.

\bibitem{kong1994}
{\sc Kong, A., Liu, J. and Wong, W.} (1994).
\newblock Sequential imputations and {B}ayesian missing data problems.
\newblock {\em J. Amer. Statis. Assoc.\/} {\bf 89,} 278--288.

\bibitem{lambert2008}
{\sc Lambert, A.} (2008).
\newblock Population dynamics and random genealogies.
\newblock {\em Stochastic models\/} {\bf 24,} 45--163.

\bibitem{linetal2013}
{\sc Lin, M., Chen, R. and Liu, J.~S.} (2013).
\newblock Lookahead strategies for {Sequential Monte Carlo}.
\newblock {\em Statistical Science\/} {\bf 28,} 69--94.

\bibitem{liu1995}
{\sc Liu, J. and Chen, R.} (1995).
\newblock Blind deconvolution via sequential imputations.
\newblock {\em J. Amer. Statist. Assoc.\/} {\bf 90,} 567--576.

\bibitem{liu1998}
{\sc Liu, J. and Chen, R.} (1998).
\newblock Sequential {M}onte {C}arlo methods for dynamic systems.
\newblock {\em J. Amer. Statist. Assoc.\/} {\bf 93,} 1032--1044.

\bibitem{meleard2012}
{\sc M\'{e}l\'{e}ard, S. and Villemonais, D.} (2012).
\newblock Quasi-stationary distributions and population processes.
\newblock {\em Probability Surveys\/} {\bf 9,} 340--410.

\bibitem{nasell1999}
{\sc Nasell, I.} (1999).
\newblock On the quasi-stationary distribution of the stochastic logistic
  epidemic.
\newblock {\em Math. Biosci.\/} {\bf 156,} 21--40.

\bibitem{neal2014}
{\sc Neal, P. et~al.} (2014).
\newblock Endemic behaviour of {SIS} epidemics with general infectious period
  distributions.
\newblock {\em Adv. in App. Prob.\/} {\bf 46,} 241--255.

\bibitem{pollock2016scalable}
{\sc Pollock, M., Fearnhead, P., Johansen, A. and Roberts, G.} (2016).
\newblock The scalable langevin exact algorithm: Bayesian inference for big
  data.
\newblock {\em arXiv preprint arXiv:1609.03436\/}.

\bibitem{vandoorn2008}
{\sc van Doorn, E. and Pollett, P.} (2008).
\newblock Survival in a quasi-death process.
\newblock {\em Lin. Alg. and App.\/} {\bf 429,} 776--791.

\bibitem{vandoorn2013}
{\sc van Doorn, E. and Pollett, P.} (2013).
\newblock Quasi-stationary distributions for discrete-state models.
\newblock {\em Euro. J. Op. Res.\/} {\bf 230,} 1--14.

\bibitem{vandoorn2014}
{\sc van Doorn, E.~A.} (2014).
\newblock Representations for the decay parameter of a birth-death process
  based on the {C}ourant-{F}ischer theorem.

\bibitem{wri:1931}
{\sc Wright, S.} (1931).
\newblock Evolution in {Mendelian} populations.
\newblock {\em Genetics\/} {\bf 16,} 97--159.

\end{thebibliography}

\end{document}